\documentclass[letterpaper, 10 pt, conference]{ieeeconf}
\IEEEoverridecommandlockouts
\usepackage{cite}
\usepackage{amsmath,amssymb,amsfonts}
\usepackage{algorithmic}
\usepackage{amsthm}
\usepackage{graphicx}
\usepackage{textcomp}
\usepackage{xcolor}
\usepackage{soul}
\usepackage{caption}
\usepackage{subcaption}
\def\BibTeX{{\rm B\kern-.05em{\sc i\kern-.025em b}\kern-.08em
    T\kern-.1667em\lower.7ex\hbox{E}\kern-.125emX}}
    
    \newcommand{\R}{\mathbb{R}}

\newcommand{\mcl}[1]{\mathcal{#1}}

\newcommand{\N}{\mathbb{N}}

\newcommand{\eps}{\varepsilon}

\newtheorem{defn}{Definition}
\newtheorem{thm}{Theorem}
\newtheorem{prop}{Proposition}

\begin{document}
% \title{Estimating the Region of Attraction Using\\
 \title{Combining Trajectory Data with Analytical Lyapunov Functions for Improved Region of Attraction Estimation %\textbf{CURRENT VERSION
 %\title{Combining Trajectory Data with Energy Functions for Improved Region of Attraction Estimation \textbf{CURRENT VERSION
 %}\\
% \title{Combining Converse Lyapunov Approximation with Energy Functions for Improved Region of Attraction Estimation\\
% \title{Transient Stability Region of Attraction Estimation using Synchrophasors and Bernstein Polynomials\\
%\title{A Two Step Approach for Estimation of the Region of Attraction of Transient Stability of a Generator\\
%{\footnotesize \textsuperscript{*}Note: Sub-titles are not captured in Xplore %and
%should not be used}
% \thanks{This work is supported by the Brazilian agencies CNPQ (grant
% 170100/2018-9), São Paulo Research Foundation (FAPESP) (grants
% 2016/08645-9, 2018/07375-3, 2018/20104-9 and 2019/10033-0), Engie (grant PD-00403-0053) and NSF CMMI-1933243.}
}
\vspace{-0.1cm}
\author{Lucas Lugnani$^{1}$;  Morgan Jones$^{2}$; Luís F. C. Alberto$^{3}$; Matthew Peet$^{4}$ and Daniel Dotta$^{1}$% <-this % stops a space
\thanks{*This work is supported by the Brazilian agencies CNPQ (grant
170100/2018-9), São Paulo Research Foundation (FAPESP) (grants
2016/08645-9, 2018/07375-3, 2018/20104-9 and 2019/10033-0), Engie (grant PD-00403-0053) and NSF CMMI-1933243.}% <-this % stops a space
\thanks{$^{1}$Lucas Lugnani and Daniel Dotta are with Faculty of Electrical and Computer Engineering, University of Campinas, Brazil
        {\tt\small lugnani@dsee.fee.unicamp.br; dottad@unicamp.br}}%
\thanks{$^{2}$Morgan Jones is with the Department of Automatic Control and Systems Engineering, University of Sheffield, UK
        {\tt\small morgan.jones@sheffield.ac.uk; }}%
\thanks{$^{3}$Luís F. C. Alberto is with the São Carlos Engineering School, University of São Paulo, São Carlos, Brazil
        {\tt\small lfcalberto@usp.br}}% 
\thanks{$^{4}$Matthew Peet is with the School for the Engineering of Matter, Transport and Energy, Arizona State University, Tempe, USA
        {\tt\small mpeet@asu.edu}}%        
}
\vspace{-0.3cm}
% \author{\IEEEauthorblockN{1\textsuperscript{st} Lucas Lugnani}
% \IEEEauthorblockA{%\textit{School of Electrical and Computer Engineering (FEEC)} \\
% \textit{ \small University of Campinas}\\ \small
% Campinas, Brazil \\
% lugnani@dsee.fee.unicamp.br}
% \and
% \IEEEauthorblockN{2\textsuperscript{nd} Morgan Jones}
% \IEEEauthorblockA{%\textit{School for the Engineering of Matter, Transport and Energy} \\
% \textit{\small Arizona State University}\\
% \small Tempe, USA \\
% }
% %mcjone19@asu.edu}
% \and
% \IEEEauthorblockN{3\textsuperscript{rd} Luís F. C. Alberto}
% \IEEEauthorblockA{%\textit{dept. name of organization (of Aff.)} \\
% \textit{ \small University of São Paulo}\\ \small
% São Carlos, Brazil \\
% lfcalberto@usp.br}
% \and
% \IEEEauthorblockN{4\textsuperscript{th} Matthew Peet}
% \IEEEauthorblockA{%\textit{dept. name of organization (of Aff.)} \\
% \textit{ \small Arizona State University}\\ \small
% Tempe, USA \\
% mpeet@asu.edu}
% \and
% \IEEEauthorblockN{5\textsuperscript{th} Daniel Dotta}
% \IEEEauthorblockA{%\textit{dept. name of organization (of Aff.)} \\
% \textit{ \small University of Campinas}\\ \small
% Campinas, Brazil \\
% dottad@unicamp.br}
% }

\maketitle
\thispagestyle{empty}
\pagestyle{empty}%{plain}

\begin{abstract}
The increasing uptake of inverter based resources (IBRs) has resulted in many new challenges for power system operators around the world. The high level of complexity of IBR generators makes accurate classical model-based stability analysis a difficult task. This paper proposes a novel methodology for solving the problem of estimating the Region of Attraction (ROA) of a nonlinear system by combining classical model based methods with modern data driven methods. Our method yields certifiable inner approximations of the ROA, typical to that of model based methods, but also harnesses trajectory data to yield an improved accurate ROA estimation. The method is carried out by using analytical Lyapunov functions, such as energy functions, in combination with data that is used to fit a converse Lyapunov function. Our methodology is independent of the function fitting method used. In this work, for implementation purposes, we use Bernstein polynomials to function fit. Several numerical examples of ROA estimation are provided, including the Single Machine Infinite Bus (SMIB) system, a three machine system and the Van-der-Pol system.
\end{abstract}

%%%%% \begin{IEEEkeywords}
%%%%% WAMS, Region of Attraction, Bernstein polynomials, Lyapunov Stability.
%%%%% \end{IEEEkeywords}

\vspace{-0.3cm}
\section{Introduction}
\vspace{-0.2cm}
Stability analysis is of uttermost importance for the secure planning and operation of modern power systems. Of particular interest is the Transient Angular Stability of a system, defined as the ability of the system to maintain rotor angle synchronism following a disturbance and its subsequent angle excursion~\cite{kundur1994power}. Given a  Stable Equilibrium Point (SEP), the Region of Attraction (ROA), defined as the set of initial conditions for which the system tends to the SEP, provides a metric of the strength of angular stability of the system. Moreover, knowledge of the ROA can be used to provide protection parameters and limits of operation that maintain the stability and safety of the system. 
 
For general nonlinear systems, which is the case for generators connected to the grid, there does not exist an analytical expression for the ROA. In the absence of an analytical expression, there is a need for methods that can compute approximations of the ROA. Classically, methods that approximate the ROA in power systems rely on precise system models of generators and line admittances~\cite{moon2000estimating,jin2005power}. However, the increasing penetration of IBRs has resulted in more complex machine models and more dynamic operating points. For such complex systems it has become increasingly intractable to use classical model-based methods to accurately approximate the ROA~\cite{8444083}. Fortunately, the advent of Wide Area Measurement Systems (WAMS), gathering high-frequency synchrophasor data has provided new sets of system data with minimal model dependency. Some works have already explored the use of synchrophasors for ROA estimation in the literature. %, with some \textcolor{red}{sensible contributions.[Seems a vague thing to say]}

%\textcolor{red}{[Any short comings of these papers that our paper is better at? Maybe shorten the detail of explaining the methods of the other papers to get more space? r: MY POINTING OF SHORT COMINGS IS IN THE PARAGRAPH ``ALTHOUGH THIS METHODOLOGY...''. I DON'T MENTION THIS TWO WORKS BECAUSE THEY ARE CLASSICAL BUT RATHER COMMENT THE SHORTCOMINGS OF THE MORE MODERN PAPERS.]}

%\textcolor{red}{[Possible thing to include in the intro: Why cant we use SOS (non-poly vector field). Energy functions alone are Conservative. Most other data based methods, like monte carlo simulation, cannot provide a certifiable inner approx of the ROA. I MENTION THIS AT THE LAST PARAGRAPH.]}

In~\cite{Aranya07}, data driven ROA estimation is realized through the application of energy function analysis, using PMU measurements that monitor tie-lines of dynamic power flows. Authors from \cite{8725530} propose an alternative method that uses Global Phase Portraits (GPPs) that contain the singularity points at infinity, providing bounds on the basins of attraction of attractor sets. In a similar vein to the data-driven methods proposed in this paper, authors from~\cite{zhai2021estimating,colbert2018using,lai2021nonlinear} use measurement data from stable trajectories to approximate converse Maximal Lyapunov Functions (LFs) and hence construct ROA estimations. Although these methodologies bring new capabilities for high dimensional stability analysis, these methods do not guarantee an inner approximation of the ROA, unlike more classical model-based methods. Notable model based methods include~\cite{anghel2013algorithmic} where the stability analysis of power systems is analyzed by constructing LFs using Sum-of-Squares (SOS) programming. Since power systems have nonlinear trigonometric terms, non-automated algebraic reconfiguration is required to use SOS. Alternatively, the works of \cite{8725530,8444083} make analytical approaches that improve upon classical energy function based methods. 

Unfortunately, it is often intractable to compute accurate ROA estimations of power systems using model based methods. On the other hand, although data based methods can provide accurate ROA estimations, they do not yield LFs and hence cannot certify inner ROA approximations. The goal of this work is to bridge the gap between the model and data based methods to yield accurate inner ROA approximations.

%Most importantly, there is no appropriate discussion of the use of the derivative of the LF in measurement-based methods for the analysis of the power system stability, nor a discussion of how an inner approximation of the ROA using trajectory data can be achieved.

The main contribution of this paper, presented in Thm.~\ref{thm: two step LF}, shows how the existence of two functions, $V_1$ and $V_2$, provides a certifiable inner approximation of the ROA of a given ODE. Specifically, if $V_2$ is a LF and $V_1$ (not necessarily a LF) is decreasing along the solution map inside a ``donut"-shaped region, $\{ \hspace{-0.05cm} x  \hspace{-0.05cm} \in \hspace{-0.05cm} D \hspace{-0.05cm} : \hspace{-0.05cm} \gamma_1 \hspace{-0.1cm} \le \hspace{-0.05cm} V_1(x) \hspace{-0.1cm} \le \hspace{-0.05cm}  \gamma_2 \hspace{-0.05cm} \}$, we show that it is possible to construct an improved ROA estimation, as compared with the ROA approximation yielded by the LF, $V_2$, alone. For implementation, we find such a $V_1$ by function fitting a converse LF using trajectory data. %For implementation purposes, this function fitting is realized through Bernstein polynomial approximation.
% Note, there are many numerical methods that can compute such a $V_1$ and $V_2$. In this paper, for implementation, we have chosen to find such a $V_1$ by approximating a converse LF from trajectory data using Bernstein polynomial approximations, and have chosen to find $V_2$ by considering energy functions (known to give conservative ROA approximations). Using this implementation, and through application of Theorem~\ref{thm: two step LF}, we have numerically demonstrated that trajectory data can successfully be used to improve the ROA approximations yielded by energy functions. However, we would like to emphasize that Theorem~\ref{thm: two step LF} is independent of how $V_1$ and $V_2$ is found, and thus may be of broader interest. The computation of $V_1$ is particularly well suited to model-free data-based methods since it is not required to satisfy the rigid conditions necessary of being a LF (positive definiteness and decreasing along the solution map for all points other than the SEP). 

\vspace{-0.3cm}
\section{Notation} \label{sec:notation}
\vspace{-0.2cm}
%For a set $X \subset \R^n$ we say $x \in X$ is an interior point of $X$ if there exists $\eps>0$ such that $\{y \in \R^n: ||x-y||< \eps\}\subset X$. We denote the set of all interior points of $X$ by $X^\circ$. The point $x \in X$ is a limit point of $X$ if for all $ \eps>0$ there exists $ z \in \{y \in \R^n / \{x\}: ||x-y||<\eps\}$ such that $z \in X$; we denote the set of all limit points of $X$, called the closure of $X$, as $(X)^{cl}$. We say a set $X \subset \R^n$ is closed if $X=(X)^{cl}$.
We denote the $\eta>0$ neighborhood of a set $S \subset \R^n$ as $B_\eta(S):=\{ y \in \R^n: \inf_{x \in S} ||x-y||_2< \eta\}$, where $||\cdot||_2$ is the euclidean norm. In the case $S=\{x\}$ then $B_\eta(x)$ becomes a ball of radius $\eta>0$ centered at $x \in \R^n$. We denote the set of all interior points of $S \subset \R^n$ by $S^\circ$.  Let $C(\Omega, \Theta)$ be the set of continuous functions with domain $\Omega \subset \R^n$ and image $\Theta \subset \R^m$. For $\alpha \in \N^n$ we denote the partial derivative $D^\alpha := \Pi_{i=1}^{n} \frac{\partial^{\alpha_i} }{\partial x_i^{\alpha_i}} $ where by convention if $\alpha=[0,..,0]^\top$ we denote $D^\alpha f(x):=f(x)$ for any function $f$. We denote the set of $i$'th continuously differentiable functions by $C^i(\Omega,\Theta):=\{f \in C(\Omega,\Theta): D^\alpha f \in C(\Omega, \Theta) \text{ } \text{ for all } \alpha \in \N^n \text{ such that } \sum_{j=1}^{n} \alpha_j \le i\}$. For $V \in C^1(\R^n, \R)$ we denote $\nabla V:= (\frac{\partial V}{\partial x_1},....,\frac{\partial V}{\partial x_n})^\top$. We denote the space of $d$-degree polynomials $p: \Omega \to \Theta$ by $\mcl{P}_d(\Omega,\Theta)$.

\vspace{-0.2cm}
\section{Stability of ODEs} \label{sec:ODE}
\vspace{-0.1cm}
%\subsection{Ordinary Differential Equations}
\vspace{-0.1cm}
Consider a dynamical system, represented by a nonlinear ordinary differential equation (ODE) of the form
\vspace{-0.2cm}
\begin{equation} \label{eqn: ODE}
\dot{x}(t) = f(x(t)), \quad x(0)=x_0\in \R^n, \quad t \in [0,\infty)
\end{equation}
where $f: \R^n \to \R^n$ is the vector field and $x_0 \in \R^n$ is the initial condition. WLOG throughout this paper we will assume {$f(0)=0$} so the origin is an equilibrium point; a linear coordinate transformation can always be used to shift any equilibrium point to the origin.

    For simplicity in the following we assume Eq.~\eqref{eqn: ODE} is well defined. That is there exists a unique solution map $\phi_f \in C^1( \R^n \times \R^+, \R^n)$ that satisfies $\frac{\delta \phi_f(x,t)}{\delta t}= f(\phi_f(x,t))$, $\phi_f(x,0)=x$ and $\phi_f(\phi_f(x,t),s)=\phi_f(x,t+s)$. Sufficient conditions for the existence and uniqueness of a solution map, based on the smoothness properties of the vector field, can be found in standard  textbooks such as~\cite{khalil2002nonlinear}.
        Given an ODE~\eqref{eqn: ODE}, we next introduce notions of asymptotic and exponential stability that are important in showing the existence of the converse LF given later in Eq~\eqref{eq: Zubov LF}.
    \vspace{-0.2cm}
    \begin{defn}
       The equilibrium point $x=0$ of ODE~\eqref{eqn: ODE} is,
\begin{itemize}
    \item stable if, for each $\eps>0$, there exists $\delta>0$ such that 
    \begin{align} \nonumber
 ||\phi_f(x,t)||_2<\eps \text{ for all } x \in B_\delta(0) \text{ and } t \ge 0.  \end{align}
        \item asymptotically stable if it is stable and there exists $\delta>0$ such that $\lim_{t \to \infty} ||\phi_f(x,t)||_2=0 \text{ for all } x \in B_\delta(0).$
%   \begin{align} \nonumber
%         \lim_{t \to \infty} ||\phi_f(x,t)||_2=0 \text{ for all } x \in B_\delta(0).
%   \end{align}
   \item exponentially stable if there exists $\lambda,\mu>0$ such that 
      \begin{align} \nonumber
         ||\phi_f(x,t)||_2<\mu e^{-\lambda t} ||x||_2 \text{ for all } x \in B_\delta(0) \text{ and } t \ge 0.
   \end{align}
\end{itemize}
    \end{defn}
    % \begin{defn}
    %     Consider an ODE~\eqref{eqn: ODE} defined by some vector field $f$. We say a set $U \subset \R^n$ is asymptotically stable if 1) $U$ contains a neighborhood of the origin. 2) $U$ is invariant. That is for all $x \in U$ we have that $\phi_f(x,t) \in U$ for all $t \ge 0$. 3) For all $x \in U$ we have that $\lim_{t \to \infty} ||\phi_f(x,t)||_2 =0$.
        
    %     \textcolor{blue}{Furthermore, we say the set $U$ is exponentially stable if there exists $\mu, \delta>0$ such $||\phi_f(x,t)||_2 < \mu||x||_2 e^{-\delta t}$ for all $t \ge 0$ for all $x \in U$.}
    % \end{defn}
%A system given by an ODE~\eqref{eqn: ODE} is said to be locally asymptotically stable if it has a non-empty asymptotically stable set. The system is said to be locally exponentially stable if it has a non-empty asymptotically stable set.

For a given asymptotically stable ODE~\eqref{eqn: ODE} the main aim of this paper is to estimate the Region of Attraction (ROA):
\vspace{-0.5cm}
\begin{align} \label{eq:ROA}
ROA_f:=\{x\in R^n : \lim_{t \to \infty} ||\phi_f(x,t)||_2 = 0\}.    
\end{align}

%\noindent called the Region of Attraction (ROA). Intuitively the ROA is the set of all initial conditions that result in the solution map converging to the origin and is hence the union of all asymtotically stable sets of the ODE.% As we have assumed $f(0)=0$, so the origin is an equilibrium point, we have that $0 \in ROA_f$ making $ROA_f \ne \emptyset$.

 There is no universal method for analytically solving nonlinear ODEs. Thus, over the years, arguably the most commonly used method to estimate $ROA_f$ is Lyapunov's second method that indirectly estimates $ROA_f$ using Lyapunov Functions (LFs); functions that are globally non-negative that decrease along the solution map. The following theorem shows how the sublevel set of a LF can approximate $ROA_f$. In order to present the main Lyapunov theorem used in this paper we recall the definition of an invariant set.
 \begin{defn}
     A set $S \subset \R^n$ is an invariant set of ODE~\eqref{eqn: ODE} if for all $x \in S$ we have $\phi_f(x,t) \in S$ for all $t \ge 0$.
 \end{defn}
\vspace{-0.3cm}
\begin{thm}[LaSalle's Invariance Principle~\cite{khalil2002nonlinear}] \label{thm: LF subset of ROA}
Consider an ODE~\eqref{eqn: ODE} defined by some vector field $f \in C^1(\R^n, \R^n)$. Suppose there exits $V \in C^1(D, \R)$ and a compact { invariant set} $S \subseteq D$ such that
\vspace{-0.3cm}
\begin{align*}
 %   &\text{If } x \in S \text{ then } \phi_f(x,t) \in S \text{ for all } t\ge 0,\\
    & \nabla V(x)^\top f(x)  \le 0 \text{ for all } x \in S.
\end{align*}
Let $E:=\{x \in S: \nabla V(x)^\top f(x)=0\}$. Then for all $x \in S$ and $\eps>0$ there exists $T>0$ such that $\phi_f(x,t) \in {B_\eps(E)}$. 

Furthermore, if $0 \in D$, $V(0)=0$, $V(x)>0$ for all $x \in D/\{0\}$ and $\phi_f(x,t) \in E$ for all $t \ge 0$ iff $x=0$ then the ODE is asymptotically stable.

Moreover, if $\gamma>0$ is such that $\{x\in S: V(x) \le  \gamma\} \subseteq  {S} $ then $\{x\in S: V(x) \le \gamma\} \subseteq ROA_f$.
\end{thm}
Theorem~\ref{thm: LF subset of ROA} shows that for a given ODE, if we can find a LF, then we can construct an inner-approximate of the ROA of the ODE. However, this theorem does not show that there must necessarily exists a LF for a given ODE or that the ROA of the ODE can be exactly characterized by an LF. 

% In~\cite{vannelli1985maximal} it was shown that for locally asymptotically stable ODEs there exists a LF of the form,
% \vspace{-0.1cm} \begin{align} \label{eq: LF max}
%     W(x):=\int_{0}^\infty \alpha(||\phi_f(x,t)||_2) dt,
% \end{align}
% %\vspace{-0.1cm}
% where $\alpha: [0,\infty) \to [0,\infty)$ is some monotonically increasing function with $\alpha(0)=0$. The converse LF given in Eq.~\eqref{eq: LF max} is called the maximal LF because it has the property that its $\infty$-sublevel set is equal to $ROA_f$, that is $\{x \in \R^n: W(x)<\infty\} = ROA_f$.

%Unfortunately, the maximal LF, given in Eq.~\eqref{eq: LF max}, is unbounded outside of $ROA_f$ and therefore is not continuous over compact sets that contain points outside of $ROA_f$. This makes approximating maximal LFs challenging. Fortunately, 

It has been shown in~\cite{jones2021converse} that for any locally exponentially stable ODE~\eqref{eqn: ODE}, there exists a bounded and continuous LF of the form,
\vspace{-0.2cm} \begin{align} \label{eq: Zubov LF}
    V^*_{\lambda, \beta}(x) \hspace{-0.1cm}: = \begin{cases} \hspace{-0.05cm} 1 \hspace{-0.05cm} - \hspace{-0.05cm} \exp \left( \hspace{-0.05cm} - \hspace{-0.05cm} \lambda \hspace{-0.05cm} \int_{0}^\infty ||\phi_f(x,t)||_2^{2\beta} dt \right) \hspace{-0.1cm} \text{ if } x \in ROA_f \\
    \hspace{-0.05cm} 1 \text{ otherwise,} \end{cases}
\end{align}
where $\lambda>0$ and $\beta \in \N$. Moreover, for sufficiently large $\lambda$ and $\beta$, this converse LF, $V^*_{\lambda, \beta}$, is Lipschitz continuous and hence differentiable almost everywhere (by Rademacher's theorem). The smoothness properties of this particular converse LF makes it highly suitable for function fitting. Furthermore, $\{ x \in \R^n: V^*_{\lambda, \beta}(x) <1 \} = ROA_f$.
\vspace{-0.2cm}
\section{Fitting Bernstein polynomials to converse Lyapunov functions} \label{sec: fitting Bernstein poly to LF}
\vspace{-0.2cm}
Given an ODE~\eqref{eqn: ODE}, in this section we show that by using an ODE solver to generate trajectory data,  Eq.~\eqref{eq: Zubov LF} can be used to construct input-output data of a converse LF. By fitting a function to this data we can approximate this converse LF in the hope of constructing an ROA estimation.% This converse LF has the property that its $1$-sublevel set is equal to the ROA of the ODE. Thus by examining the $1$-sublevel sets of our approximation we can hope to construct an estimation of the ROA of the ODE.

Specifically, we fit polynomial functions to the generated input/output data of the converse LF. Although there are many ways to fit polynomials to data (each having their relative advantages and disadvantages), in this paper, we have chosen a method based on Bernstein approximations. As we will next see, Bernstein's method for fitting polynomials to data is an optimization free approach that is guaranteed to converge uniformly.
\vspace{-0.2cm}
\subsection{Bernstein Approximation of Smooth Functions}
\vspace{-0.2cm}
We now provide a brief description of how Bernstein polynomials can approximate smooth functions. For a more in-depth overview of the field we refer to~\cite{veretennikov2016partial}. Now, recalling from Section~\ref{sec:notation} that we defined $\mcl P_d(\R^n, \R)$ as the set of $d$-degree polynomials we next define the Bernstein operator.
\vspace{-0.2cm}
\begin{defn}
	We denote the degree $d \in \N$ Bernstein operator by $\mcl B_d: C(\R^n, \R) \to \mcl P_{d }(\R^n, \R)$ and for $V \in C(\R^n, \R)$ we define $\mcl B_d V \in \mcl P_{d }(\R^n, \R)$ by
	\begin{align} \label{eq: Bernstein}
\vspace{-20pt}
%falied attempt t put on one line
% 	\mcl B_d V (x) & \hspace{-0.1cm} := \hspace{-0.15cm} \sum_{k_n=0}^d \hspace{-0.15cm} \cdots \hspace{-0.15cm} \sum_{k_{1}=0}^d \hspace{-0.1cm} V(k_1/d,....,k_n/d)  \hspace{-0.1cm} \prod_{j=1}^{n} \hspace{-0.1cm} {d \choose k_j} x_j^{k_j} (1- x_j)^{d - k_j}.
	\mcl B_d V (x) &  := \sum_{k_n=0}^d  \cdots  \sum_{k_{1}=0}^d  V(k_1/d,....,k_n/d)  \\ \nonumber
	& \qquad \times \prod_{j=1}^{n}  {d \choose k_j} x_j^{k_j} (1- x_j)^{d - k_j}.
	\end{align}
\end{defn}
\noindent Given a function $V \in C(\R^n, \R)$, we can calculate the polynomial $\mcl B_d V$ using Eq.~\eqref{eq: Bernstein} with only knowledge of the values of $V$ at uniformly gridded points in $[0,1]^n$. Thus, in order to calculate $\mcl B_d V$ it is not necessary to have an analytic expression of $V$. We next recall that $\mcl B_d V \to V$ uniformly as $d \to \infty$. Moreover, although the Bernstein approximation in Eq.~\eqref{eq: Bernstein} only involves the value of the $0'th$ differential order of $V$, if $V$ is differentiable, then it follows that the derivative of $\mcl B_d V$ will also converge to the derivative of $V$. This is a particularly useful when it comes to approximating converse LFs because we would also like our approximation to be a LF itself. Thus to make our approximation, $P$, have the property that it decreases along solution trajectories, $\dot{P}(x)<0$, we ensure the derivative of $P$ also approximates the derivative of the converse LF, $\dot{V}(x)<0$. 
\vspace{-0.2cm}
\begin{thm}[Multivariate uniform approximation by Bernstein polynomials, see Theorem~4 in~\cite{veretennikov2016partial}] \label{thm: bernstein approx}
	Given $\alpha=(\alpha_1,..., \alpha_n) \subset \N^n$ suppose $D^{\alpha}V \in C(\R^n, \R)$ then it follows 
	\begin{align}
	\vspace{-20pt}
	\lim_{d \to \infty} \sup_{x \in [0,1]^n}| D^{\alpha}\mcl{B}_dV(x) - D^{\alpha}V(x)|=0.
	\end{align}
%	Moreover if $V \in LocLip(\R^n, \R)$ with Lipschitz constant $L_V>0$ we have that
%	\begin{align}
%	\sup_{x \in [0,1]^n}| \mcl{B}_dV(x) - V(x)|< \frac{1+L_V}{d^{\frac{1}{3}}}.
%	\end{align}
\end{thm}
%\paragraph{Bernstein approximation over 
Theorem~\ref{thm: bernstein approx} shows that Eq.~\eqref{eq: Bernstein} can be used to approximate functions over $[0,1]^n$. Note, using the same methodology, we may also approximate a function, $V$, over some set $[a,b]^n$ where $a<b$. In order to do this we first apply a linear coordinate change mapping $[a,b]^n$ to $[0,1]^n$, defining  $\tilde{V}(x):= V\left(\frac{x_1-a_1}{b_1-a_1},...,\frac{x_n-a_n}{b_n-a_n}\right)$.We then apply Eq.~\eqref{eq: Bernstein} to approximate $\tilde{V}(x)$ over $[0,1]^n$, yielding $\mcl B_d \tilde{V}$ such that $	\lim_{d \to \infty} \sup_{x \in [0,1]^n}| D^{\alpha}\mcl{B}_d \tilde{V}(x) - D^{\alpha} \tilde{V}(x)|=0$ (by Theorem~\ref{thm: bernstein approx}). Finally, we again change the coordinates, mapping $[0,1]^n$ back to $[a,b]^n$, defining $J(x): =\mcl B_d \tilde{V}((b_1-a_1){x_1} + a_1,....,(b_n-a_n){x_n} + a_n)$. It then follows that
\begin{align} \label{eq: fitting on rectable}
&\lim_{d \to \infty} \sup_{x \in [a,b]^n}| D^{\alpha}J(x) - D^{\alpha} {V}(x)|\\ \nonumber
&= \hspace{-0.08cm} \lim_{d \to \infty} \hspace{-0.05cm} \sup_{x \in [0,1]^n} \hspace{-0.05cm} \left| D^{\alpha}J\left(\frac{x_1-a_1}{b_1-a_1},...,\frac{x_n-a_n}{b_n-a_n} \right) \hspace{-0.05cm} - \hspace{-0.05cm} D^{\alpha} \tilde{V}(x) \right|\\ \nonumber
& = \lim_{d \to \infty} \sup_{x \in [0,1]^n}| D^{\alpha}\mcl{B}_d \tilde{V}(x) - D^{\alpha} \tilde{V}(x)|=0.
\end{align}

\subsection{Generating Data from Converse Lyapunov Functions}
Converse LFs can be approximated by polynomials using Eq.~\eqref{eq: Bernstein}. However, in order to apply Eq.~\eqref{eq: Bernstein} we must know the value of the converse LF at uniformly gridded points in $[a,b]^n$ (note approximation over $[a,b]^n$ rather than $[0,1]^n$ can be achieved through linear coordinate transformations, see Eq.~\eqref{eq: fitting on rectable}). %\textcolor{red}{In  this section we show that the values of converse Lyapunov functions can be found using trajectory data.}

Given a set of initial conditions, $\{x_i\}_{1 \le i \le N} \subset [a,b]^n$, and a terminal trajectory time $T>0$, it is possible to generate trajectory data $D_{i,j}:= ||\phi_f(x_i,(j-1) \Delta t)||_2$, where $\Delta t >0$ is some small time-step and $1 \le j \le \frac{T+1}{\Delta t}$. This can be achieved using any ODE solver, for instance, Matlab's \texttt{ODE45}. Of course, in order to use an ODE solver this does require complete knowledge of the vector field, $f$. In the case where the model of the system is unknown it may still be possible to generate the required trajectory data from experimental data. Through the semi-group property of solution maps, $\phi_f(\phi_f(x,t),s)=\phi_f(x,t+s)$, knowledge of just a single trajectory can generate a vast number of data points. Moreover, interpolation can be used in places where there are gaps in our data knowledge. 

Now, given $\lambda>0$, $\beta \in \N$ trajectory data, $D \in \R^{N \times (K+1)}$, for sufficiently large $K \in \N$, we can approximate the value of the corresponding converse LF given in Eq.~\eqref{eq: Zubov LF} in the following way
\begin{align} \label{eqn: traj data value of LF}
  & V^*_{\lambda,\beta}(x_i) \approx 1 - e^{-\lambda W(x_i)},\\ \nonumber
    & \text{ where } W(x_i) \approx \int_0^{K \Delta t} ||\phi_f(x_i,t)||_2^{2 \beta} dt \approx \sum_{j=1}^{K+1} D_{i,j}^{2 \beta} \Delta t.
\end{align}

%\section{A Two Step Approach to ROA Approximation}
\section{Improving ROA Estimation with Approximated Converse Lyapunov Functions} \label{sec: Improving ROA Estimation with Approximated Converse Lyapunov Functions}
Given an ODE defined by some vector field $f$, we have shown that through applications of Eqs.~\eqref{eq: Bernstein} and~\eqref{eqn: traj data value of LF} that it is possible to numerically \st{a} construct a Bernstein polynomial approximation, $\mcl B_d V^*_{\lambda, \beta}$ for some $d \in \N$, $\lambda>0$ and $\beta \in \N$, of the converse LF, $V^*_{\lambda, \beta}$, given in Eq.~\eqref{eq: Zubov LF}. Moreover, assuming that $D^\alpha V^*_{\lambda, \beta}$ is continuous, where $\alpha \in \N^n$, Theorem~\ref{thm: bernstein approx} can be used to show that $\lim_{d \to \infty} D^\alpha \mcl B_d V^*_{\lambda, \beta} \to D^\alpha V^*_{\lambda, \beta}$ uniformly in $[a,b]^n$ (note approximation over $[a,b]^n$ rather than $[0,1]^n$ can be achieved through linear coordinate transformations, see Eq.~\eqref{eq: fitting on rectable}).

Ideally, for sufficiently large $d \in \N$, $\lambda>0$ and $\beta \in \N$, our approximation, $\mcl B_d V^*_{\lambda, \beta}$, will also be a LF over some set containing the origin; that is $\mcl B_d V^*_{\lambda, \beta}(0)=0$, $\mcl B_d V^*_{\lambda, \beta}(x) \ge 0$ for all $x \in \Omega / \{0\}$, and $\nabla \mcl B_d V^*_{\lambda, \beta}(x)^\top f(x) < 0$ for all $x \in \Omega /\{0\}$. Then Theorem~\ref{thm: LF subset of ROA} can be used to show that the sublevel set of $\mcl B_d V^*_{\lambda, \beta}$ yields an inner approximation of the ROA of the given ODE.

Unfortunately, despite the fact $\mcl B_d V^*_{\lambda, \beta}$ tends to $V^*_{\lambda, \beta}$ as $d \to \infty$ and $V^*_{\lambda, \beta}$ is a LF, $\mcl B_d V^*_{\lambda, \beta}$ is not necessarily a LF. To see this we first note that, for a given $x_0 \in [a,b]^n$ if $\nabla V^*_{\lambda, \beta}(x_0)^\top f(x_0)< -a$, where $a>0$, it follows by Theorem~\ref{thm: bernstein approx} that there exists $D \in \N$ such that $||\nabla \mcl B_d V^*_{\lambda, \beta}(x_0) -  \nabla V^*_{\lambda, \beta}(x_0)||_2 < \frac{a}{2 ||f(x_0)||_2}$ for all $d>D$. Thus, using the Cauchy Swarz inequality we have that for all $d>D$
\begin{align} \label{eq: Berstein strictly decreasing}
   & \nabla \mcl B_d V^*_{\lambda, \beta}(x_0)^\top f(x_0)  \\ \nonumber
   &= \hspace{-0.1cm} \nabla \hspace{-0.05cm} \mcl B_d V^*_{\lambda, \beta}( \hspace{-0.02cm} x_0 \hspace{-0.02cm})^\top \hspace{-0.15cm} f( \hspace{-0.02cm} x_0 \hspace{-0.02cm} ) \hspace{-0.1cm}- \hspace{-0.1cm} \nabla \hspace{-0.05cm}  V^*_{\lambda, \beta}(\hspace{-0.02cm} x_0 \hspace{-0.02cm})^\top \hspace{-0.15cm} f( \hspace{-0.02cm} x_0 \hspace{-0.02cm}) \hspace{-0.1cm} + \hspace{-0.1cm}  \nabla \hspace{-0.05cm}  V^*_{\lambda, \beta}( \hspace{-0.02cm} x_0\hspace{-0.02cm} )^\top \hspace{-0.15cm} f(\hspace{-0.02cm} x_0 \hspace{-0.02cm})\\ \nonumber
   &  \le ||\nabla \hspace{-0.05cm}  \mcl B_d V^*_{\lambda, \beta}( \hspace{-0.02cm} x_0 \hspace{-0.02cm} ) \hspace{-0.05cm} - \hspace{-0.05cm}   \nabla \hspace{-0.05cm}  V^*_{\lambda, \beta}( \hspace{-0.02cm} x_0 \hspace{-0.02cm} )||_2 ||f( \hspace{-0.02cm} x_0 \hspace{-0.02cm})||_2 \hspace{-0.05cm}  - \hspace{-0.05cm}  a \le -\frac{a}{2} <0.
\end{align}
Eq.~\eqref{eq: Berstein strictly decreasing} shows that for sufficiently large $d$, whenever $V^*_{\lambda, \beta}$ is strictly decreasing along the solution map we also have that $\mcl B_d V^*_{\lambda, \beta}$ is strictly decreasing along the solution map. However, at the origin $V^*_{\lambda, \beta}$ is not strictly decreasing along the solution map, that is $\nabla V^*_{\lambda, \beta}(0)^\top f(0)=0$. Because of this fact and the fact $\mcl B_d V^*_{\lambda, \beta}$ is continuous, in general for some finite $d \in \N$  our approximation will be such that that $\nabla \mcl B_d V^*_{\lambda, \beta}(x)^\top f(x) \ge 0$ for all $x$ in some small neighborhood of the origin. Thus in general, for finite $d \in \N$, it follows that we will have $B_d V^*_{\lambda, \beta}(x)^\top f(x)<0$ for all $x$ in some ``donut shaped" region, $\{y \in \R^n: \gamma_1 \le B_d V^*_{\lambda, \beta}(y) \le \gamma_2\}$ for some $\gamma_1<\gamma_2$, as opposed to a sublevel set $\{y \in \R^n: B_d V^*(y) \le \gamma_2\}$. By a similar argument, in general, we do not expect  $B_d V^*_{\lambda, \beta}(0)=0$ for any fixed $d \in \N$ and thus in general $B_d V^*_{\lambda, \beta}$ is not a LF.

Although, $B_d V^*_{\lambda, \beta}$ is not a LF, and therefore cannot certify the origin is asymptotically stable, we will next show, in Prop.~\ref{prop: donut LF B_eta attractor set}, that functions strictly decreasing along the solution map inside some ``donut shaped" region can still be used to certify that the solution map must enter some ball of radius $\eta>0$ around the origin, $B_\eta(0)$. %Later we will combine this result with a Lyapunov function that shows that $B_\eta(0)$ is asymptotically stable
\vspace{-0.2cm}
\begin{prop} \label{prop: donut LF B_eta attractor set}
Consider an ODE~\eqref{eqn: ODE} defined by some vector field $f \in C^1(\R^n, \R^n)$ and compact set $D \subset \R^n$. Suppose $V \in C^1(D,\R)$,  $\gamma_1<\gamma_2$ and $\eta,\eps,a>0$ are such that
\begin{itemize}
\item $B_\eps(\{y \in D: V(y) \le \gamma_1\}) \subset B_\eta(0)$.
\item $ \{y \in D: V(y) \le \gamma_2\} \subset {D^\circ}$
    \item For all $x \in D/\{y \in D: V(y) \le \gamma_1\}$ we have that
    \begin{equation} \label{pfeq: V deacreasing along traj}
        \nabla V(x)^\top f(x)<-a<0.
    \end{equation}
\end{itemize}
 Then it follows that for any $x \in \{y \in D: V(y) \le \gamma_2\}/B_\eta(0)$ there exists $T>0$ such that
\begin{equation}
    \phi_f(x,T) \in B_\eta(0). 
\end{equation} 
\vspace{-25pt}
\end{prop}
\begin{proof} Let $S:= \{y \in D:  V(y) \le \gamma_2\}$ and $\tilde{V}(x):=\rho(V(x))$ where $\rho$ is an infinitely smooth function defined by,
\begin{align*}
    \rho(x):=\begin{cases} \exp\left(\frac{-1}{(\gamma_1-x)^2} \right) \text{ for } x >\gamma_1,\\
    0 \text{ otherwise.}
    \end{cases}
    \vspace{-25pt}
\end{align*}
Now, it is clear that $\tilde{V}(x)=0$ for all $x \in \{y \in D:  V(y) \le \gamma_1\}$ and hence $\nabla \tilde{V}(x)^\top f(x)=0$ for all $x \in \{y \in D:  V(y) \le \gamma_1\}$. On the other hand, for $x \in S / D$ we have that
\begin{align*}
   \nabla \tilde{V}(x)^\top f(x)= \frac{2\rho(V(x)) \nabla V(x)^\top f(x)}{( V(x)-\gamma_1)^3}<0.
   \vspace{-25pt}
\end{align*}
Therefore, $\nabla \tilde{V}(x)^\top f(x) \le 0$ for all $x \in S$. Moreover, since $S\subset D^\circ$ and $V(\phi_f(x,t))$ is strictly decreasing on $\partial D$ it follows that $S$ is an invariant set. We are now in a position to apply Thm.~\ref{thm: LF subset of ROA} for $\tilde{V}$. 

Let $E:=\{y \in D: \nabla \tilde{V}(x)^\top f(x)=0\}$. Clearly from the definition of $\rho$ it follows that $E=\{y \in D: {V}(x)\le \gamma_1\}$ and hence $B_\eps(E) \subset B_\eta(0)$.

Now, for $x \in S$ and $\delta<\eps$ by Thm.~\ref{thm: LF subset of ROA} there exists $T>0$ such that $\phi_f(x,T) \in B_\delta(E) \subset B_\eps(E) \subset B_\eta(0)$. $\blacksquare$
\end{proof}
In Prop.~\ref{prop: donut LF B_eta attractor set} we have proposed conditions, based on some function we denote here as $V_1$, that certify that the solution map of a given ODE must enter some ball, $B_\eta(0)$. Recall that $V_1$ can be found by approximating a converse LF by Bernstein polynomials.

We next show that if there exists a LF, $V_2$, that certifies that $B_\eta(0)$ is an asymptotically stable set, then $V_1$ and $V_2$ can be used together to provide an improved inner approximation of the ROA of the given ODE. %Note that the LF, $V_2$, may be found by considering energy functions, linearizing the ODE and computing quadratic Lyapunov functions, or any other method that yields a local Lyapunov function.
\vspace{-0.2cm}
\begin{thm} \label{thm: two step LF}
Consider an ODE~\eqref{eqn: ODE} defined by some vector field $f \in C^1(\R^n, \R^n)$ and compact set $D \subset \R^n$. Suppose there exists $V_1,V_2 \in C^1(D,\R)$,  $\gamma_1<\gamma_2$ and $\eta,\eps,a>0$ are such that
\begin{itemize}
\item $B_\eps(\{y \in D: V(y) \le \gamma_1\}) \subset B_\eta(0)$.
\item $ \{y \in D: V(y) \le \gamma_2\} \subset {D^\circ}$
    \item For all $x \in D/\{y \in D: V(y) \le \gamma_1\}$ we have that
    \begin{equation} \label{pfeqthm: V deacreasing along traj}
        \nabla V(x)^\top f(x)<-a<0.
    \end{equation}
\end{itemize}
Moreover, suppose $0 \in D$ and for some $\gamma_3>0$ $V_2$ satisfies 
    \begin{itemize}
\item $V_2(0) = 0$ and $V_2(x)>0$ for all $ x \in D/\{0\}$.
\item $\nabla V_2(x)^\top f(x) \le 0$ for all $ x \in D$.
\item \hspace{-0.1cm}$   \phi_f(x,t) \hspace{-0.12cm} \in \hspace{-0.1cm}  \{ x \hspace{-0.05cm} \in \hspace{-0.05cm} D \hspace{-0.05cm} : \hspace{-0.05cm} \nabla V_2(x)^\top \hspace{-0.1cm}  f(x) \hspace{-0.05cm} = \hspace{-0.05cm} 0\} \hspace{-0.07cm} \text{ for } \hspace{-0.05cm} t \hspace{-0.05cm} \ge \hspace{-0.05cm} 0 \hspace{-0.05cm} \text{ iff } \hspace{-0.05cm} x \hspace{-0.05cm}= \hspace{-0.05cm}0.$
\item $B_\eta(0) \subseteq \{y \in D: V_2(y) \le \gamma_3\} \subseteq D^\circ$.
\end{itemize}
Then $\{y \in D: V_1(y) \le \gamma_2\} \cup \{y \in D: V_2(y) \le \gamma_3\}  \subseteq ROA_f$.
\end{thm}
\vspace{-0.2cm}
\begin{proof}
By Theorem~\ref{thm: LF subset of ROA} it follows that $\{y \in D: V_2(y) \le \gamma_3\}\subseteq ROA_f$.

If $x \in \{y \in D: V_1(y) \le \gamma_2\}$ then by Prop.~\ref{prop: donut LF B_eta attractor set} there exists $T>0$ such that
\vspace{-0.2cm}
\begin{equation}
    z:=\phi_f(x,T) \in B_\eta(0) \subseteq \{y \in D: V_2(y) \le \gamma_3\}. 
\end{equation}
\vspace{-0.1cm}
Since $\{y \in D: V_2(y) \le \gamma_3\}\subseteq ROA_f$ it follows that $\lim_{t \to \infty} ||\phi_f(z,t)||_2 =0$. Therefore, using the semi-group properties of solution maps we have that 
\vspace{-0.2cm}
\begin{align*}
    &\lim_{t \to \infty} ||\phi_f(x,t)||_2 = \lim_{t \to \infty} ||\phi_f( \phi_f(x,T),t-T)||_2\\
    &\quad = \lim_{t \to \infty} ||\phi_f( z,t-T)||_2= \lim_{s \to \infty} ||\phi_f( z,s)||_2=0,
\end{align*}
implying $x \in ROA_f$. Because the same argument can be used for any $x \in \{y \in D: V_1(y) \le \gamma_2\}$ it follows that $\{y \in D: V_1(y) \le \gamma_2\} \subseteq ROA_f$.

Since $\{y \in D: V_2(y) \le \gamma_3\}\subseteq ROA_f$ and $\{y \in D: V_1(y) \le \gamma_2\} \subseteq ROA_f$ it follows $\{y \in D: V_1(y) \le \gamma_2\} \cup \{y \in D: V_2(y) \le \gamma_3\}  \subseteq ROA_f$. $\blacksquare$
\end{proof}

\textbf{\underline{Practical Implementation of Theorem~\ref{thm: two step LF} }}\\
Thm.~\ref{thm: two step LF} shows how two functions, $V_1$ and $V_2$, can be used to produce an inner approximation of the ROA of a given ODE. The function, $V_2$, is a classical LF (positive semidefinite and decreasing along trajectories) but provides a poor estimation of the ROA, $\{y \in D: V_2(y) \le \gamma_3\} \subseteq ROA_f$. On the other hand, $V_1$ need not be a classical LF as there is no requirement that $V_2(0)=0$ or $\nabla V_2(x)^\top f(x) \le 0$ near the origin, making the computational search for a candidate $V_1$ amenable to data based methods. Whereas, the computational search for a candidate $V_2$ is redistricted to model based methods (energy functions or Sum-of-Squares programming).

In this paper we compute a candidate $V_1$ by fitting a Bernstein polynomial to a converse LF (Section~\ref{sec: fitting Bernstein poly to LF}). We compute candidate a $V_2$ using several methods including: 
\begin{itemize}
    \item Linearizing the system and then computing a quadratic LF by solving the resulting Linear Matrix Inequalities (LMIs) (see Subsection~\ref{subsec: Van Poll}).
    \item Using energy functions (see Subsection~\ref{subsec: SMIB}).
    \item Using Sum-of-Square (SOS) programming (see Subsection~\ref{subsec: model A}).
\end{itemize} 
\noindent Now, for a given ODE defined by a vector field $f$ and candidate functions, $V_1$ and $V_2$, we next outline how to compute $\eta,\gamma_1,\gamma_2,\gamma_3 \in \R$ from Thm.~\ref{thm: two step LF} to estimate $ROA_f$.\\
\textit{Step 1:} Compute the largest ROA estimation yielded by the LF $V_2$ using Thm.~\ref{thm: LF subset of ROA}. This amounts to finding $\{x \in D: V_2(x) \le \gamma^*_3\} \subseteq ROA_{f}$ where
\begin{align} \label{opt: 1 gam_3}
    &\gamma^*_3 \hspace{-0.05cm} \in \hspace{-0.05cm}  \arg \sup_{\gamma \in \R} \{\gamma\} \hspace{-0.05cm} \text{ such that } \hspace{-0.05cm}  \{ \hspace{-0.05cm}  x \hspace{-0.05cm} \in \hspace{-0.05cm} D \hspace{-0.05cm} : \hspace{-0.05cm} V_2(x) \hspace{-0.05cm} \le \hspace{-0.05cm} \gamma\} \hspace{-0.05cm} \subseteq \hspace{-0.05cm} S_{V_2},
\end{align}
and $S_{V_2} := \{x \in D: V_2(x)>0\} \cap \{x \in D: \nabla V_2(x)^T f(x) \le 0 \}$.\\
\textit{Step 2:} Find the largest ball that is certified to be contained inside the ROA. That is, solve 
\begin{align} \label{opt: 2 eta}
    \eta^* \in \arg \sup_{\eta>0}\{\eta\} \text{ such that } B_\eta(0)\subseteq \{x \hspace{-0.05cm} \in \hspace{-0.05cm} D: V_2(x) \hspace{-0.05cm} \le \hspace{-0.05cm} \gamma^*_3\}
\end{align}
\textit{Step 3:} Find largest sublevel set of $V_1$ contained in $B_\eta(0)$. That is, solve 
\begin{align} \label{opt: 3 gam_1}
\gamma_1^*\in \arg \sup_{\gamma>0} \{\gamma\} \text{ such that } \{x \hspace{-0.05cm} \in \hspace{-0.05cm} D: V_1(x) \hspace{-0.05cm} < \hspace{-0.05cm} \gamma\} \hspace{-0.05cm} \subset \hspace{-0.05cm} B_\eta(0).
\end{align}
\textit{Step 4:} Find the largest "donut shaped" set such that $V_1$ is decreasing. That is, solve
\begin{align} \label{opt: 4 gam_2}
& \gamma_2^*\in  \arg \sup_{\gamma \in \R}\{\gamma\} \text{ such that } \\ \nonumber
& \{ \hspace{-0.05cm} x \hspace{-0.05cm} \in \hspace{-0.05cm}  D: \gamma_1^* \le \hspace{-0.05cm} V_1(x) \hspace{-0.05cm} \le \gamma \hspace{-0.05cm} \} \subseteq \{ \hspace{-0.05cm} x \hspace{-0.05cm} \in \hspace{-0.05cm} D: \nabla V_1(x)^\top f(x) \hspace{-0.05cm} < \hspace{-0.05cm}0 \hspace{-0.05cm} \}.
\end{align}
Opts.~\eqref{opt: 1 gam_3}, \eqref{opt: 2 eta} \eqref{opt: 3 gam_1} and~\eqref{opt: 4 gam_2} are all set containment problems that can be solved using Putinar's Positivstellensatz to formulate auxiliary Sum-of-Squares optimization problems (possibly using algebriac constraints to enforce non-polynomial terms such as in~\cite{anghel2013algorithmic}). Alternatively, for two and three dimensional systems we can solve these Opts graphically by preforming bisection on $\eta,\gamma_1,\gamma_2,\gamma_3 \in \R$, plotting and verifying the set containment's hold. In the same way we can attempt to approximately solve these Opts for higher dimensional systems by plotting slices of the state space and graphically certifying the set containments. 

\begin{figure*}[t!]
        \begin{subfigure}[b]{0.25\textwidth}
              \includegraphics[width=\linewidth , trim = {1.75cm 1cm 1.75cm 1cm}, clip]{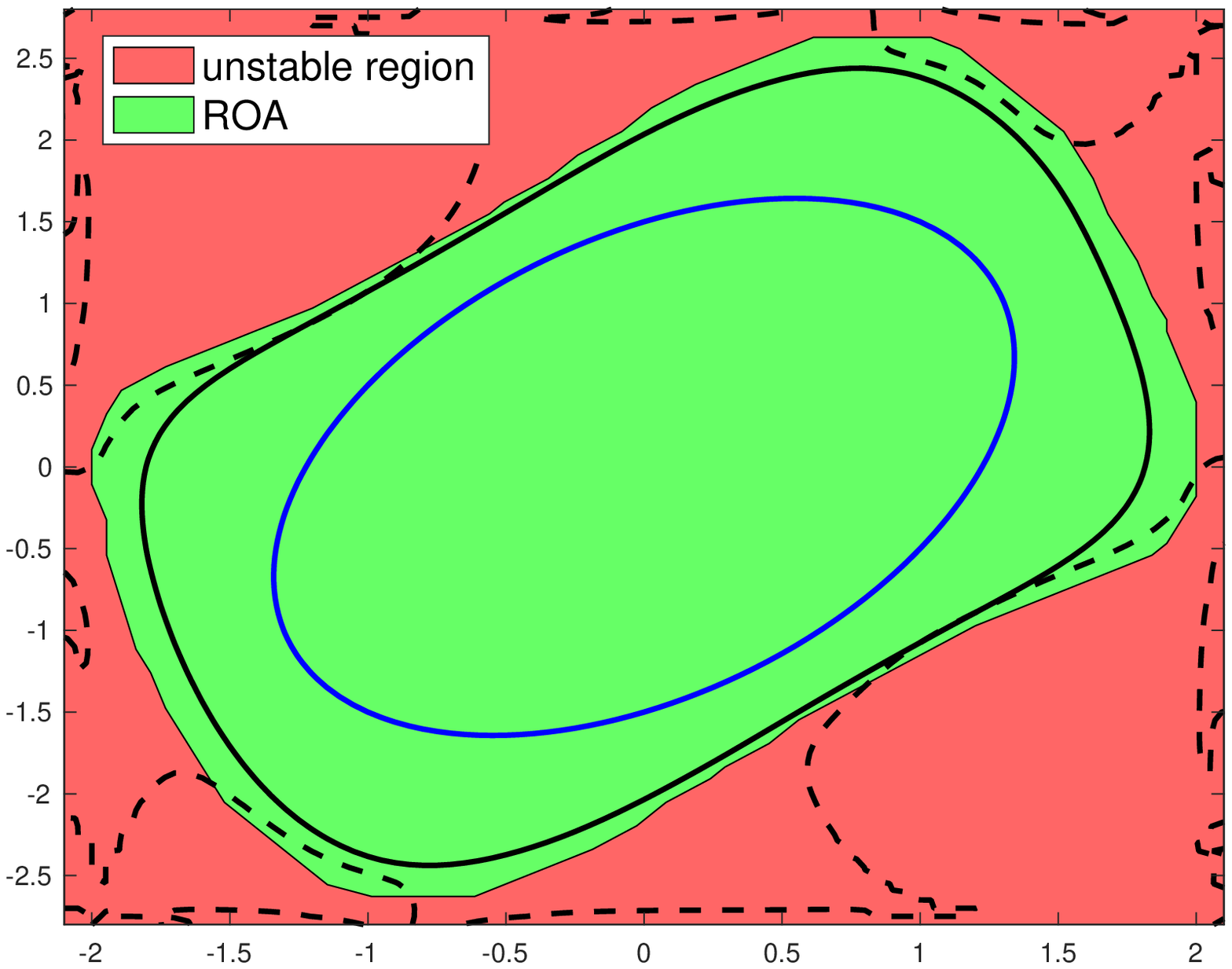}
               \vspace{-10pt}
                \caption{The Van der Pol system}
                \label{fig:deg_75_van_poll}
        \end{subfigure}%
        \begin{subfigure}[b]{0.25\textwidth}
                \includegraphics[width=\linewidth, trim = {1.75cm 1cm 1.75cm 1cm}, clip]{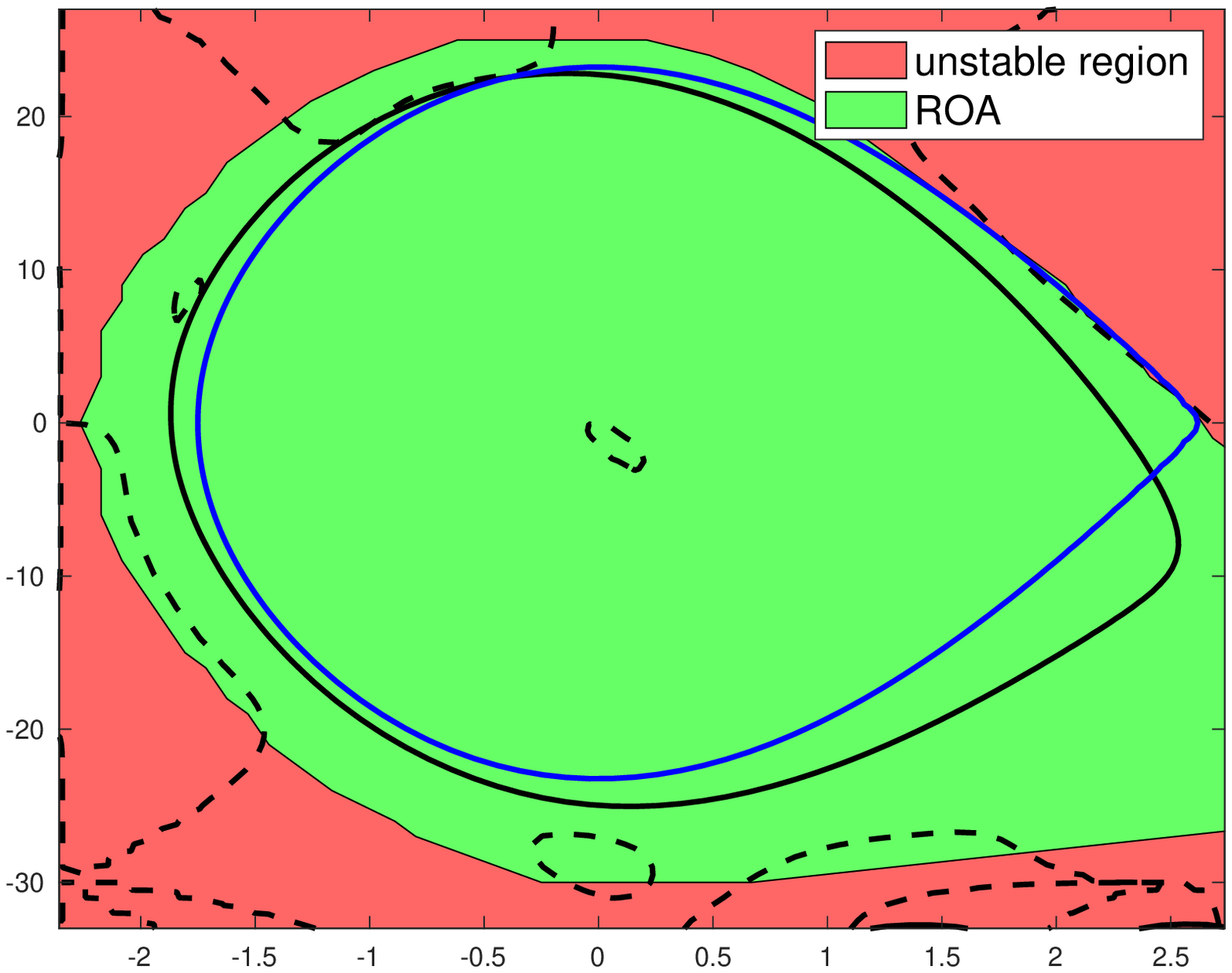}
                 \vspace{-10pt}
                \caption{The SMIB system}
                \label{fig:BernsteinROAsublvl}
        \end{subfigure}%
        \begin{subfigure}[b]{0.25\textwidth}
                \includegraphics[width=\linewidth, trim = {1.75cm 1cm 1.75cm 1cm}, clip]{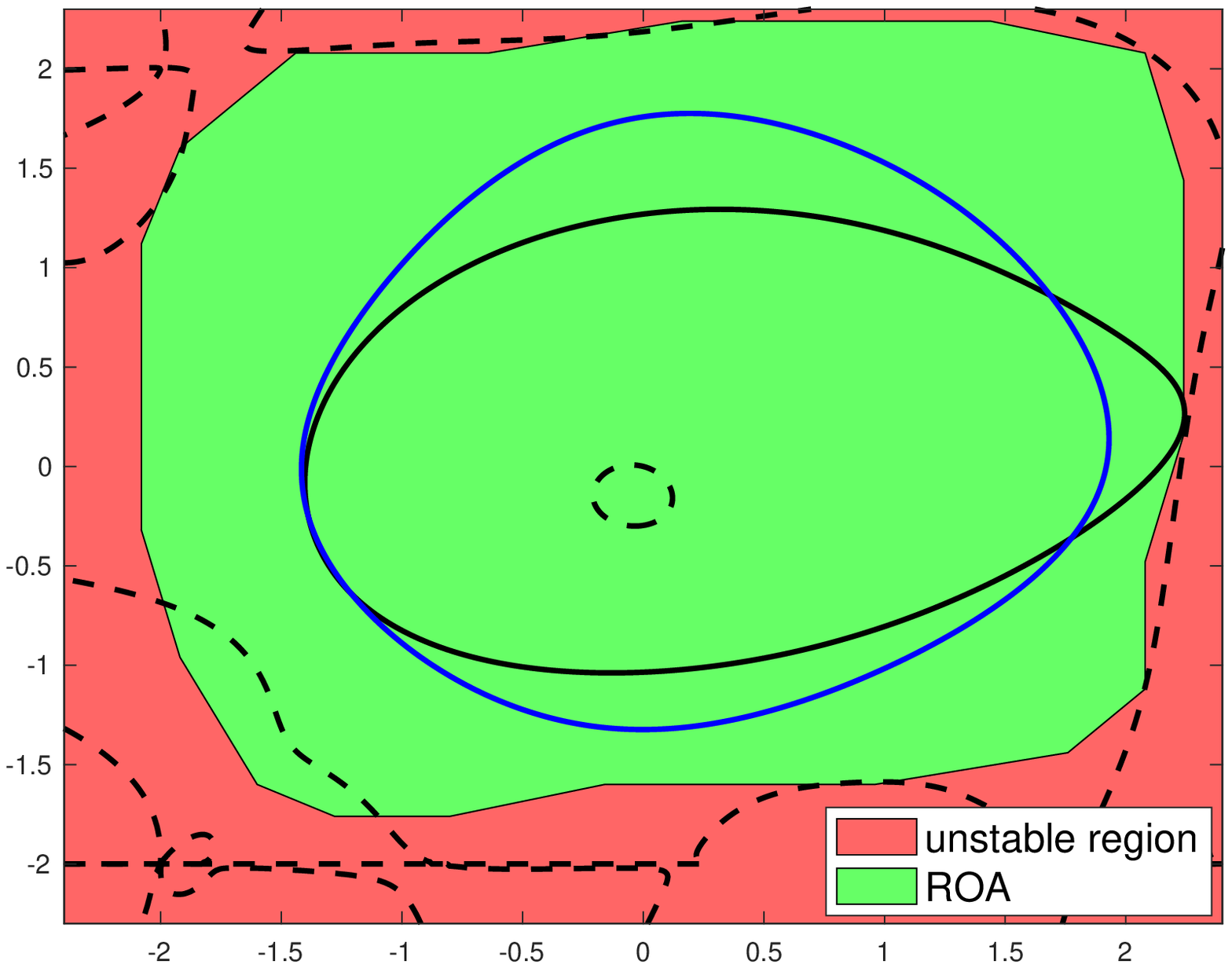}
                \vspace{-10pt}
                \caption{A two machine system}
                \label{fig:Model B}
        \end{subfigure}%
        \begin{subfigure}[b]{0.25\textwidth}
         \includegraphics[width=\linewidth , trim = {1.75cm 1cm 1.75cm 1cm}, clip]{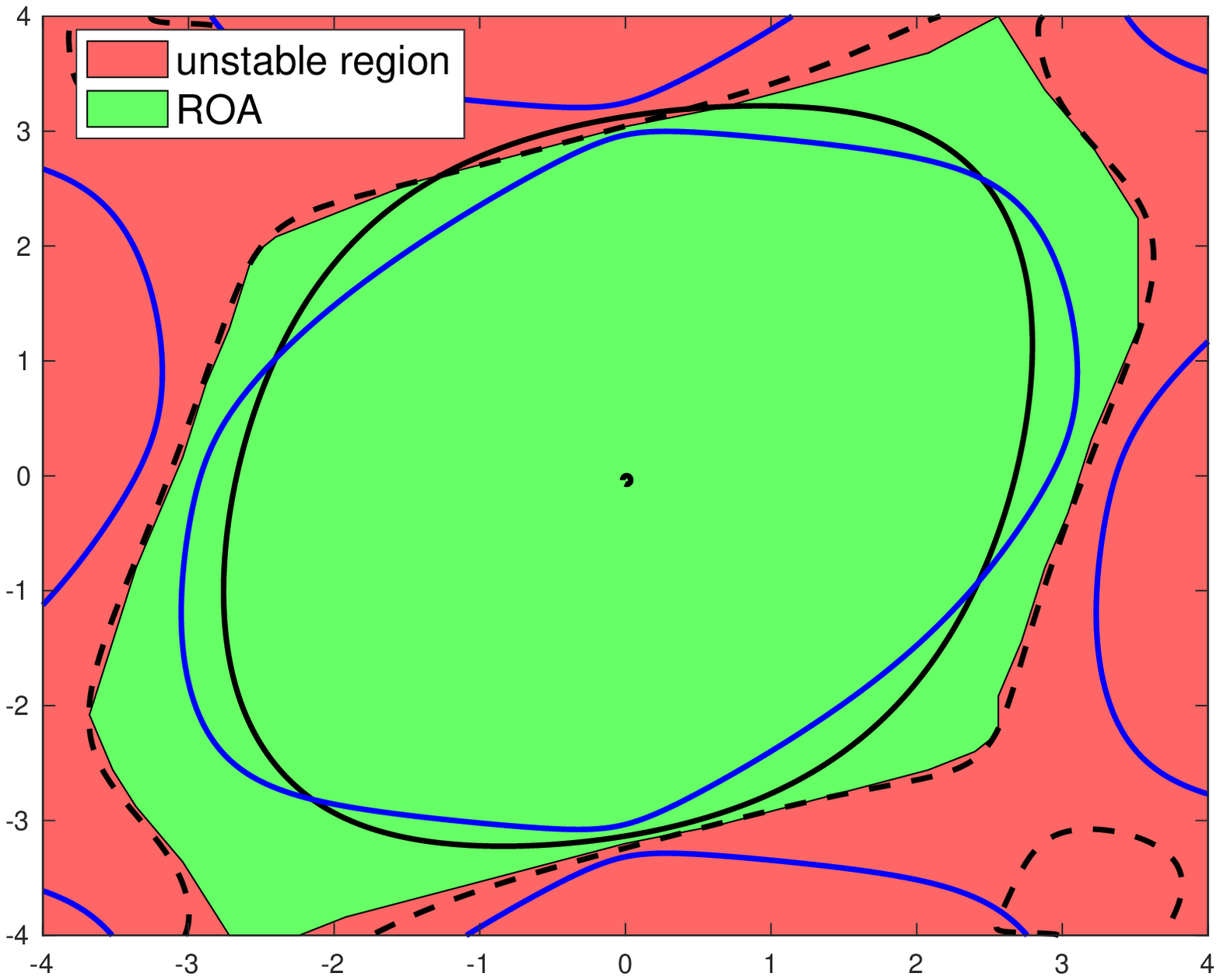}
          \vspace{-10pt}
                \caption{IEEE three machine system}
                \label{fig:Model A}
        \end{subfigure}
        \vspace{-20pt}
        \caption{\footnotesize Our estimations of the ROA of several systems, given by the union of the region contained inside the black ($V_1$) and blue ($V_2$) curves. The dotted line corresponds to the boundary of the $0$-sublevel set of the derivative of $V_1$ along the solution map.}\label{fig:ROAs}
        \vspace{-20pt}
\end{figure*}

\vspace{-0.2cm}
\section{Numerical Examples}
\vspace{-0.2cm}
We next present several numerical examples demonstrating how Thm.~\ref{thm: two step LF} can be used to yield accurate ROA approximations of nonlinear ODEs. For each numerical example we compute $V_1$ and $V_2$, from Thm.~\ref{thm: two step LF}, by respectively fitting Bernstein polynomials to the converse LF given in Eq.~\eqref{eq: Zubov LF} (using Eqs.~\eqref{eq: Bernstein}, \eqref{eq: fitting on rectable} and~\eqref{eqn: traj data value of LF}) and either using an energy function or an analytical LF (found previously in the literature).  

To demonstrate the accuracy of our approximations of the ROA we will carry out extensive Monte Carlo simulations of the solution map to estimate the stable and unstable regions in each figure. Although this Monte Carlo method can estimate the ROA well it does not account for simulation error or provide a LF and hence cannot provide a certified ROA inner approximation.

\subsection{Estimating the ROA of the Van der Pol system} \label{subsec: Van Poll}
\vspace{-0.1cm}
Consider the reverse time Van der Pol oscillator defined by the vector field:
\vspace{-0.4cm}
\begin{align} \label{eqn: van der pol ode}
%\dot{x}_1(t) & = -x_2(t)\\ \nonumber
%\dot{x}_2(t) & = x_1(t) - x_2(t)(1- x_1^2(t)).
f_{VDP}(x)=\begin{bmatrix} -x_2 \\ x_1-x_2(1-x_1^2) \end{bmatrix}
\end{align}
The following quadratic LF was found in~\cite{khalil2002nonlinear}: $V_{VDP}(x)=x^\top P x$, where $P=\begin{bmatrix} 1.5 & -0.5 \\ -0.5 & 1 \end{bmatrix}$. 

We now fit the converse LF, $V^*_{\lambda,\beta}$ (given in Eq.~\eqref{eq: Zubov LF}), for $\lambda=3$ and $\beta=1$ by a degree $75$ Bernstein polynomial over the set $D=[-2, 2] \times [-2.7, 2.7]$. In Fig.~\ref{fig:deg_75_van_poll} we have plotted our estimation of the ROA achieved using this fitted Bernstein polynomial as $V_1$ and $V_{VDP}$ as $V_2$ in Thm.~\ref{thm: two step LF}. The black and blue curves correspond to the boundaries of $\{x \in D: V_1(x)< 0.74\}$ and $\{x \in D: V_2(x)< 2.25\}$ respectively.

% \begin{figure}[!t]
%     \centering
%     \includegraphics[width=\columnwidth, trim = {0cm 0cm 0cm 0cm}, clip]{MC_V1V2_ROA_VDP.eps}
%     \vspace{-24pt}
%   \caption{~}
%     \label{fig:deg_75_van_poll}
%     \vspace{-10pt}
% \end{figure}
\vspace{-0.1cm}
\subsection{Estimating the ROA of the Single Machine Infinite Bus (SMIB) system} \label{subsec: SMIB}
The SMIB system can be modeled by an ODE~\eqref{eqn: ODE} with the following vector field:
% Consider the SMIB system defined by the following ODE:
% \vspace{-0.45cm}
% \begin{align} \label{ODE: swing system}
%     & \dot{\delta}(t)=\omega(t) \\ \nonumber
%     & 2H\dot{\omega}(t)=P_{m} - \frac{E'E_{B}}{X_{eq}}sin(\delta(t)) - \Bar{D}\omega(t),
% \end{align}
% where $P_m, E',E_B,X_{eq},D \in \R$ \textcolor{red}{(Is $\bar{D}$ different from $D$, $P_m$ not in SMIB dynamics?)} are known constants. The equilibrium point of this system occurs at $[\delta_{ep},0]^T \in \R^2$, where $\delta_{ep}:=\sin^{-1}\left(\frac{P_{m}X_{eq}}{E'E_{B}} \right)$. To be consistent with Section~\ref{sec:ODE}, we make the coordinate change \st{${\delta}= \tilde{\delta} - \delta_{ep}$}\textcolor{blue}{${\tilde{\delta}}= \delta - \delta_{ep}$}, mapping the equilibrium point $[\delta_{ep},0]^T$ to the origin. Now, denoting the state as $x=[\tilde{\delta},\omega]$, the system given in Eq.~\eqref{ODE: swing system} can be written as an ODE~\eqref{eqn: ODE} with the following vector field, 
\vspace{-0.2cm}
\[f_{SMIB}(x) \hspace{-0.07cm} := \hspace{-0.07cm} \begin{bmatrix}x_2\\ (1/2H)({P_m \hspace{-0.05cm} - \hspace{-0.05cm} \frac{E'E_{B}}{X_{eq}}sin(x_1 \hspace{-0.05cm} + \hspace{-0.05cm} \delta_{ep}) \hspace{-0.05cm} - \hspace{-0.05cm} {D}x_2}) \end{bmatrix} \hspace{-0.07cm} ,\]
where $H=0.0106 s^2/rad$, $X_{eq}=0.28 pu$, $P_m=1 pu$, $E_{B}=1 pu$, $E'=1.21 pu$, $D=0.03$.   
It is shown in~\cite{chowPowerBook20} that the energy of the SMIB system can be expressed as the following function
\vspace{-0.4cm}
    \begin{align*}%\label{eqn: energyfunction}
    & V_{E}(x) \hspace{-0.1cm} := \hspace{-0.1cm}
       -P_m x_1 \hspace{-0.1cm} + \hspace{-0.1cm} \frac{E' E_B}{X_{eq}} (\cos(\delta_{ep}) \hspace{-0.1cm}-\hspace{-0.1cm} \cos(x_1 \hspace{-0.1cm}+\hspace{-0.1cm} \delta_{ep})) \hspace{-0.1cm}+\hspace{-0.1cm} H x_2^{2}.
\end{align*}
%The energy function given in Eq.~\eqref{eqn:energyfuntcion} can be used as a candidate LF. %Clearly $V(0)=0$ and it is shown in~\cite{chowPowerBook20} that $V(x)>0$ about some punctured neighborhood of the origin (excluding the point $0 \in \R^n$). Moreover,
%\begin{align*}
%& \nabla V_{E}(x)^T f_{SMIB}(x)\\
%& = \left[-P_m + \frac{E' V}{X_{eq}} \sin(x_1 + \delta_{ep} ), 2Hx_2 \right]^T f_{SMIB}(x)\\
%&= x_2 \left( -P_m + \frac{E' V}{X_{eq}} \sin(x_1 + \delta_{ep} ) \right) \\
%& \quad + x_2 \left( {P_m - \frac{E'E_{B}}{X_{eq}}sin(x_1 + \delta_{ep}) - \Bar{D}x_2} \right)\\
%& = - \Bar{D} x_2^2 \le 0 \text{ for all } x \in \R^n.
%\end{align*}
% \vspace{-0.3cm}
% \begin{align*}
% & \nabla V_{E}(x)^T f_{SMIB}(x) = - \Bar{D} x_2^2 \le 0 \text{ for all } x \in \R^n.
% \end{align*}
% It was shown in~\cite{chowPowerBook20} that $\phi_{f_{SMIB}}(x,t) \in \{x \in \R^n: x_2=0\}$ iff $x=0$. 
By graphically solving Opt.~\eqref{opt: 1 gam_3} for $V_2=V_E$ and $f=f_{SMIB}$, we find $\gamma^*_3=5.722$. The boundary of $\{x \in \R^n: V_E(x) \le \gamma^*_3\}$ is given as the blue curve in Fig.~\ref{fig:BernsteinROAsublvl} 

%We now estimate the ROA of the SMIB system given in Eq.~\eqref{ODE: swing system} with model constants: $P_m, E',E_B,X_{eq},D \in \R$: $\omega_0=1 pu$, $M=2H=0.0212 s^2/rad$, $X_{eq}=0.28 pu$, $P_m=1 pu$, $E_{B}=1 pu$, $E'=1.21 pu$, $D=1.03$. \textcolor{red}{Double check parameters. Cant seem to easily reproduce because dont know how to convert units. Is $\bar{D}$ different from $D$?}  
We find a candidate $V_1$ function of Thm.~\ref{thm: two step LF} by fitting a degree $60$ Bernstein polynomial to the converse LF, $V^*_{\lambda,\beta}$ (given in Eq.~\eqref{eq: Zubov LF}), for $\lambda=10$ and $\beta=1$ over the set $D=[-0.75\pi, \pi] \times [-30, 30]$.

 We have plotted the boundary of $\{y \in D: V_1(y) \le 0.68\}$ as the black curve in Fig.~\ref{fig:BernsteinROAsublvl}. These sublevel sets are such that Thm.~\ref{thm: two step LF} shows $\{y \in D: V_1(y) \le 0.68\} \cup \{y \in D: V_2(y) \le 5.722\}  \subseteq ROA_{f_{SMIB}}$, providing an inner approximation of $ROA_{f_{SMIB}}$. Moreover, we have plotted the boundary of the set $\{y \in D: \nabla V_1(y)^\top f_{SMIB}(y) \le 0\}$ as the dotted black line. Showing $\nabla V_1(y)^\top f_{SMIB}(y)$ may not be negative around a neighborhood of the origin as expected from Sec.~\ref{sec: Improving ROA Estimation with Approximated Converse Lyapunov Functions}.

Using both $V_{1}$ and $V_2$ we have improved the inner approximation of $ROA_{f_{SMIB}}$ as compared to the approximation of yielded by the energy function, $V_2$, alone. %Our approximation provides an almost analytical solution for the boundary of $ROA_{f_{SMIB}}$ at the top right quadrant, an important region where the system is most likely to operate in a practical situation. 

% \begin{figure}[!t]
%     \centering
%     \includegraphics[width=\columnwidth, trim = {0 0 0 0cm}, clip]{ROAsmib_astar6070.eps}
%     \vspace{-24pt}
%     \caption{Our approximation of the ROA of the SMIB is given by the union of the region contained inside the blue and black curves. The dotted blue line corresponds to the $0$-sublevel set of the derivative our Bernstein approximation of the converse LF in Eq.~\eqref{eq: Zubov LF} along the solution map.}
%     \label{fig:BernsteinROAsublvl}
%     \vspace{-10pt}
% \end{figure}

\subsection{Estimating the ROA of a two-machine versus infinite bus system}
Consider the following four dimensional power system model found in~\cite{bretas_alberto_2003,anghel2013algorithmic} that represents a two-machine versus infinite bus system which can be modeled by an ODE~\eqref{eqn: ODE} with the following vector field:
\vspace{-0.3cm}
\begin{align} \label{ODE: TMIB}
&f_{TMIB}(x)=\begin{bmatrix} f_1(x), f_2(x),f_3(x),f_4(x) \end{bmatrix}^\top,%\\ \nonumber
%&=\begin{bmatrix}
%x_2 \\
%33.5849-1.8868\cos(x_1-x_3) -5.283 \cos(x_1)-16.9811\sin(x_1-x_3)-59.6226\sin(x_1)-1.8868x_2\\
% x_4\\
%11.3924\sin(x_1-x_3)-1.2658\cos(x_1-x_3)-3.2278\cos(x_3)-1.2658x_4-99.3671\sin(x_3)+48.481
% \end{bmatrix}
\end{align} 
\noindent where $f_1(x)=x_2$, $f_2(x)=33.5849-1.8868\cos(x_1-x_3) -5.283 \cos(x_1)-16.9811\sin(x_1-x_3)-59.6226\sin(x_1)-1.8868x_2$, $f_3(x)=x_4$ and $f_4(x)=11.3924\sin(x_1-x_3)-1.2658\cos(x_1-x_3)-3.2278\cos(x_3)-1.2658x_4-99.3671\sin(x_3)+48.481$. 
The point $x_{SEP}:=[0.468, 0, 0.463, 0]^\top \in \R^4$ is a stable equilibrium point of Eq.~\eqref{ODE: TMIB}. Using a change of variables $\tilde{x}=x-x_{SEP}$ we map the equilibrium point to the origin.

We now fit the converse LF, $V^*_{\lambda,\beta}$ (given in Eq.~\eqref{eq: Zubov LF}), for $\lambda=1$ and $\beta=1$ by a degree $20$ Bernstein polynomial over the set $[-2, 2] \times [-3, 3] \times [-2, 2] \times [-3, 3]$. Fig.~\ref{fig:Model B} shows a slice of the state space when $x_2=x_4=0$ depicting the sublevel set of this Bernstein polynomial along with the ROA estimation found in~\cite{anghel2013algorithmic}. By using Thm.~\ref{thm: two step LF} we are able to certify an improved ROA estimation (the union of the sublevel sets in Fig.~\ref{fig:Model B}).

% \begin{figure}[!t]
%     \centering
%     \includegraphics[width=\columnwidth, trim = {0cm 0cm 0cm 0cm}, clip]{MC_V1V2_ROA_TMIB_deltapi_omega3.eps}
%     \vspace{-24pt}
%   \caption{~}
%     \label{fig:Model B}
%     % \vspace{-10pt}
% \end{figure}
\vspace{-0.2cm}
\subsection{Estimating the ROA of a three-machine system} \label{subsec: model A} \vspace{-0.15cm}
Consider the following four dimensional power system model found in~\cite{chiang2011direct} (Page 144) that represents a three-machine system with machine number 3 as swing bus (reference of the system) and can be modeled by an ODE~\eqref{eqn: ODE} with the following vector field:
%\begin{align}
%    \dot{x}_1(t) & =x_2(t)\\
%    \dot{x}_2(t) & =-\sin(x_1(t))-0.5\sin(x_1(t)-x_3(t))-0.4x_2(t)\\
%    \dot{x}_3(t) & = x_4(t)\\
 %   \dot{x}_4(t) & = -0.5 \sin(x_3(t)) -0.5\sin(x_3(t)-x_1(t)))-0.5x_4(t)+0.05
%\end{align}
\vspace{-0.3cm}
\begin{align} \label{ODE: 3MS}
&f_{3MS}(x) \\ \nonumber
&=\begin{bmatrix}
x_2 \\
-\sin(x_1)-0.5\sin(x_1-x_3)-0.4x_2\\
 x_4\\
-0.5 \sin(x_3) -0.5\sin(x_3-x_1)-0.5x_4+0.05 \end{bmatrix}
\end{align}
The point $x_{SEP}:=[0.02001, 0, 0.06003, 0]^\top \in \R^4$ is a stable equilibrium point of Eq.~\eqref{ODE: 3MS}. Using a change of variables $\tilde{x}=x-x_{SEP}$ we map the equilibrium point to the origin. The ROA of this system has previously been estimated using energy functions in~\cite{chiang2011direct}. A more accurate ROA estimation was found in~\cite{anghel2013algorithmic} using Sum-of-Squares to find a LF. We now use the LF found in~\cite{anghel2013algorithmic} as $V_2$ in Thm.~\ref{thm: two step LF} and compute a $V_1$ by fitting a degree $20$ Bernstein polynomial to the converse LF, $V^*_{\lambda,\beta}$ (given in Eq.~\eqref{eq: Zubov LF}), for $\lambda=1$ and $\beta=1$ over the set $[-4, 4] \times [-0.75, 0.75] \times [-4, 4] \times [-0.75, 0.75]$. Fig.~\ref{fig:Model A} shows a slice of the state space when $x_2=x_4=0$ and depicts the best ROA estimation found in~\cite{anghel2013algorithmic} along with a sublevel set of the resulting fitted Bernstein polynomial. Thm.~\ref{thm: two step LF} can be used to certify the union of these sublevel sets are inside the ROA, providing an improved ROA estimation. 
\vspace{-0.5cm}
\section{Conclusion}
\vspace{-0.2cm}
This work proposes a novel methodology for ROA estimation using an approximated converse Lyapunov function, derived from trajectory data, together with an analytical Lyapunov function. The method yields a certifiable inner ROA estimation. Numerical examples demonstrate that the proposed method is able to expand ROA approximations found using analytical Lyapunov functions derived elsewhere in the literature. This method is not limited to the converse Lyapunov function fitting technique implemented, Bernstein polynomial approximations. Function fitting techniques that are better suited for high dimensional problems will be explored in future work.
%Of particular interest is the estimation of the ROA of transient angular stability using PMU data.
%\section*{Acknowledgment}
%The authors would like to thank Brendon K. Colbert for providing base code regarding the methodology applied.
\vspace{-0.1cm}
\bibliographystyle{./bibliography/IEEEtran}
\vspace{-0.2cm}
\bibliography{./bibliography/IEEEabrv,./bibliography/IEEEexample}
% \bibliography{./bibliography/Bernstein_ROA_report.bbl}
% \bibliography{Bernstein_ROA_report.bbl}

\end{document}